\documentclass[draftclsnofoot,12pt,onecolumn]{IEEEtran}
\usepackage{amsmath, amsthm, amssymb}
\usepackage{epsfig}
\usepackage{latexsym}
\usepackage{graphicx}
\usepackage{cite}
\usepackage{subfigure}

\newtheorem{lemma}{Lemma}
\newtheorem{proposition}{Proposition}

\newtheorem{remark}{Remark}

\title{ Useful  Results for Computing the Nuttall${-}Q$ and Incomplete Toronto Special Functions}

\author{
Paschalis~C.~Sofotasios$^{1}$,~Khuong~Ho-~Van$^{2}$,~Tuan~Dang~Anh$^{2}$ and~Hung~Dinh~Quoc$^{2}$
\\
\begin{normalsize} 
$^{1}$School of Electronic and Electrical Engineering, University of Leeds, LS$2$ $9$JT, Leeds, United Kingdom. 
\end{normalsize}
\\
\begin{normalsize} 
e-mail: $\rm p.sofotasios@\rm leeds.ac.uk$
\end{normalsize} 
\\
\begin{normalsize} 
$^{2}$Department of Telecommunications Engineering, HoChiMinh City University of Technology, HoChiMinh City, Vietnam. 
\end{normalsize}
\\
\begin{normalsize} 
e-mail: $ \rm khuong.hovan@\rm yahoo.ca$, $\rm datuan@\rm hcmut.edu.vn$, $ \rm hungqd76@\rm yahoo.com$ 
\end{normalsize} 
 }

\begin{document}

\maketitle

\begin{abstract}
This work is devoted to the derivation of novel analytic results for special functions which are particularly useful in wireless communication theory. Capitalizing on recently reported series representations for the Nuttall $Q{-}$function  and the incomplete Toronto function,    we derive closed-form upper bounds for the corresponding truncation error of these series as well as closed-form upper bounds that under certain cases become accurate approximations. The derived expressions are tight and their algebraic representation is rather convenient to handle analytically and numerically. Given that the Nuttall${-}Q$ and incomplete Toronto functions are not built-in in popular mathematical software packages, the proposed results are particularly useful in computing these functions when employed in applications relating to natural sciences and engineering, such as wireless communication over fading channels.
\end{abstract}

\section{Introduction}

It is widely known that special functions constitute invaluable
mathematical tools in most fields of natural sciences and engineering. In the area
of telecommunications, their utilization often allows the derivation
of elegant closed-form representations for important performance
measures such as error probability, channel capacity and
higher-order statistics (HOS). The computational realization of such
expressions is typically straightforward since most special
functions are built-in in popular mathematical software packages
such as MAPLE, MATLAB and MATHEMATICA. Among others, the Marcum
$Q{-}$function, $Q_{m}(a,b)$, the Nuttall $Q{-}$function,
$Q_{m,n}(a,b)$ and the incomplete Toronto function (ITF),
$T_{B}(m,n,r)$  were proposed several decades ago and have been
involved in analytic solutions of various scenarios in  
information and communication theory.  More specifically, Marcum $Q{-}$function was
firstly proposed by Marcum in \cite{J:Marcum_1, J:Marcum_2} but it
became widely known in digital communications, over fading or
non-fading media, by the contributions in \cite{J:Nuttall,
B:Proakis, J:Helstrom, J:Shnidman, J:Simon_1, B:Alouini}. Its
traditional and generalized form are denoted as $Q_{1}(a, b)$ and
$Q_{m}(a, b)$, respectively, and their properties and identities
were presented in \cite{J:Nuttall}. Respective upper and lower
bounds were reported in \cite{B:Alouini, J:Karagiannidis, J:Baricz,
C:Tellambura_1} while an exponential integral representation was
proposed in \cite{C:Tellambura_1}. In addition, useful closed-form representation
for the $Q_{m}(a, b)$ function for the special case that the order $m$ is an positive odd
multiple of $0.5$, i.e. $ m \pm \frac{1}{2} \in \mathbb{N} $, were derived in \cite{J:Karagiannidis, C:Kam_1}.

Likewise, the Nuttall $Q{-}$function  is a special function that
emerges from the generalization of the Marcum $Q{-}$function
\cite{J:Nuttall}. Its definition is given by a semi-infinite
integral representation and it can be expressed in terms of
$Q_{m}(a, b)$ and the modified Bessel function of the first kind,
$I_{n} (x)$, for the special case that the sum $m+n$ is an odd
integer \cite{J:Simon_2}. Establishment of monotonicity criteria for
$Q_{m, n} (a, b)$ along with the derivation of tight lower and upper
bounds and a closed-form expression for the case that $m \pm 0.5 \in
\mathbb{N}$ and $n \pm 0.5 \in \mathbb{N} $ were reported in
\cite{J:Karagiannidis}.  In the same context, the incomplete Toronto
function is a similar special function which was proposed by Hatley
in \cite{J:Hatley}.  It  is a generalization of the  Toronto
function, $T(m,n,r)$, and includes the $Q_{m}(a,b)$ function as a
special case. Its definition is also given in non-infinite integral
form while alternative representations include two infinite series
\cite{J:Sagon}. Its application has been used in studies relating to statistics,
signal detection and estimation, radar systems and error probability
analysis \cite{J:Fisher, J:Marcum_3, J:Swerling}.

Nevertheless, in spite of the evident importance of the
$Q_{m,n}(a,b)$ and $T_{B}(m,n,r)$  functions,  they are not adequately
tabulated while they are not included as standard built-in functions
in popular software packages. Motivated by this, the Authors in
\cite{J:Karagiannidis, B:Sofotasios, C:Sofotasios_1, C:Sofotasios_2,
C:Sofotasios_3}  derived explicit expressions for the $Q_{m,n}(a,b)$
and $T_{B}(m, n, r)$ functions. Specifically, a closed-form
expression for $Q_{m,n}(a,b)$ was derived in \cite{J:Karagiannidis}
for the specific case that $m \pm \frac{1}{2} \in \mathbb{N}$ and $m
\pm \frac{1}{2} \in \mathbb{N}$. A similar expression for the
$T_{B}(m,nr)$ was subsequently derived in \cite{B:Sofotasios,
C:Sofotasios_1}. Furthermore, simple series representations for both
$Q_{m,n}(a,b)$ and $T_{B}(m,n,r)$ functions were derived in
\cite{C:Sofotasios_1, C:Sofotasios_2, C:Sofotasios_3}. These series
are particularly useful since their algebraic form is relatively
simple, which render the functions convenient to handle
analytically. However, their form is infinite and thus the
determination of adequate truncation needs to be taken into
consideration for ensuring acceptable levels of accuracy for any
given application.

Motivated by this, this work is devoted to the derivation of novel upper bounds for the truncation error of the above functions. The derived representations are expressed in closed-form and have a relatively simple algebraic form. In addition, simple upper bounds are derived in closed-form, which are shown to be quite tight and in certain cases they become remarkably accurate closed-form approximations.

\section{Novel Results for the Nuttall $Q$-function}

\subsection{Definition, Basic Properties and  Existing  Expressions}

The Nuttall Q-function is defined by the following semi-infinite integral representation \cite[eq. (86)]{J:Nuttall},
 
\begin{equation}  \label{Nuttall_1}
Q_{m, n}(a,b) \triangleq \int_{b}^{\infty} x^{m} e^{-\frac{x^2+a^2}{2}} I_{n} (ax) dx,
\end{equation}

\noindent
and constitutes a generalization of the Marcum $Q$-function,
 
\begin{equation}  \label{Marcum_1}
Q_{m}(a,b) \triangleq \frac{1}{a^{m-1}} \int_{b}^{\infty} x^{m} e^{-\frac{x^2+a^2}{2}} I_{m-1} (ax) dx.
\end{equation}

The normalized Nuttall $Q{-}$function is given by $\mathcal{Q} _{m, n} (a, b) \triangleq Q_{m, n} (a, b){/}a^{n}$, which for $n = 0$ yields $Q_{1, 0}(a, b) = \mathcal{Q} _{1, 0} (a, b) =Q_{1} (a, b) = Q(a,b)$. For the special case that $n = m - 1$ it follows that $\mathcal{Q} _{m, m-1} (a, b) = Q_{m} (a, b)$ and $Q_{m, m - 1} (a, b) = a^{m - 1} Q_{m} (a, b)$. Furthermore, when $m$ and $n$ are integers, the following recursion equation was reported in \cite[eq. (3)]{J:Simon_2},
 
\begin{equation} \label{Nuttall_Recursion}
Q_{m, n} (a, b)  = \frac{b^{m - 1}  I_{n} (ab)}{e^{\frac{a^{2} + b^{2}}{2}}} + a Q_{m-1, n+1} (a, b)      +  (m + n - 1) Q_{m - 2, n} (a, b).
\end{equation}

\noindent
A finite series representation expressed in terms of the $Q_{1}(a,b)$, $I_{n}(x)$ functions, the gamma function, $\Gamma(x)$, and the generalized $k^{th}$ order Laguerre polynomial, $L_{k}^{(n)} (x)$, was proposed in \cite[eq. (8)]{J:Simon_2}. However, this expression is restricted to the case that $m + n$ is an odd positive integer.

As already mentioned in Section I, a simple expressions for the $Q_{m,n}(a,b)$ function were reported recently by \cite{J:Karagiannidis, B:Sofotasios, C:Sofotasios_3}.   This expression is given by,
 
\begin{equation} \label{Nuttall_Polynomial}
\mathcal{Q}_{m,n}(a,b) \simeq \sum_{l = 0}^{p} \frac{a^{2l } \, e^{- \frac{a^{2}}{2}}  \Gamma(p+l) p^{1 - 2l} \, \Gamma \left( \frac{m + n + 2l + 1}{2}, \frac{b^{2}}{2} \right) }{  l! \Gamma(n+l+1) 2^{\frac{n - m + 2l + 1}{2}}\,  \Gamma(p-l+1) },
\end{equation}
where 

\begin{equation}
\Gamma(a) \triangleq \int_{0}^{\infty} t^{a - 1} \rm {exp}(-t) \, dt 
\end{equation}
and

\begin{equation}
\Gamma(a,x) \triangleq \int_{x}^{\infty} t^{a - 1} \rm{exp} (-t) \, dt
\end{equation}
denote the gamma and upper incomplete gamma functions, respectively. As, $p \rightarrow \infty$, the terms $\Gamma(p+l) p^{1 - 2l} / \Gamma(p-l+1)$ vanish and \eqref{Nuttall_Polynomial} reduces to the following exact infinite series representation,
 
\begin{equation} \label{Nuttall_Series}
\mathcal{Q}_{m,n}(a,b) =  \sum_{l = 0}^{\infty} \frac{a^{2l } \, e^{- \frac{a^{2}}{2}}   \, \Gamma \left( \frac{m + n + 2l + 1}{2}, \frac{b^{2}}{2} \right) }{  l! \Gamma(n+l+1) 2^{\frac{n - m + 2l + 1}{2}}\,   },
\end{equation}
For the special case that $m, n \in \mathbb{Z}^{+}$,  the $\Gamma(a,x)$ function can be expressed in terms of the finite series representation in \cite[eq. (8.352.4)]{B:Tables}. Hence, after necessary variable transformation \eqref{Nuttall_Polynomial} can be also expressed as,
 
\begin{equation} \label{Nuttall_Integer_Indices}
\mathcal{Q}_{m,n}(a,b) \simeq \sum_{l = 0}^{p} \sum_{k=0}^{L} \frac{\mathcal{A} \, a^{2l} b^{2k} \Gamma(p+l) p^{1 - 2l} \Gamma(\frac{m + n + 1}{2} + l)}{ l! k! \Gamma(n + l  +1) \Gamma(p-l+1) 2^{ l + k}},
\end{equation}
where,
 
\begin{equation} \label{coefficient_1}
L = \frac{m + n - 1}{2} + l,
\end{equation}
and
 
\begin{equation} \label{coefficient_2}
\mathcal{A} = a^{n}  2^{ \frac{m - n - 1}{2}} e^{- \frac{a^{2} + b^{2}}{2}}.
\end{equation}
From \eqref{Nuttall_Integer_Indices} we can also obtain straightforwardly,
 
\begin{equation} \label{Nuttall_Integer_Indices_inf}
\mathcal{Q}_{m,n}(a,b)= \sum_{l = 0}^{\infty} \sum_{k=0}^{L} \frac{\mathcal{A} \, a^{2l} b^{2k} \Gamma(\frac{m + n + 1}{2} + l)}{ l! k! \Gamma(n + l  +1)  2^{ l + k}},
\end{equation}
It is noted here that the above representations have a convenient algebraic representation and are not restricted since they hold for all values of the involved parameters. However, an efficient way that will determine effectively after how many terms the above series can be truncated without loss of accuracy, is undoubtedly necessary.

\subsection{A Closed-Form Upper Bound for the Truncation Error}

The accuracy of \eqref{Nuttall_Polynomial} is proportional to the value of $p$ and the series converges quickly. However, deriving a convenient closed-form expression that is capable of determining  the involved truncation error analytically is particularly  advantageous.

\begin{lemma}
For $m, n, a \in \mathbb{R}$ and $b \in \mathbb{R}^{+}$, the following closed-form upper bound for the truncation error holds,
 
\begin{equation*}
\epsilon_{t}  \leq \sum_{k = 0}^{\lceil n \rceil_{0.5} - 1}\frac{ (-1)^{\lceil n \rceil_{0.5}}  \Gamma(2\lceil n \rceil_{0.5} - k - 1) \mathcal{I}_{\lceil m\rceil_{0.5},\lceil n \rceil_{0.5}}^{k} (a, b) }{k! \Gamma(\lceil n \rceil_{0.5} - k) (2a)^{-k} \sqrt{\pi} 2^{ \lceil n \rceil_{0.5} - \frac{1}{2}} a^{2\lceil n \rceil_{0.5} - 1} }
\end{equation*}
\begin{equation} \label{Nuttall_Truncation}
 - \sum_{l = 0}^{p} \frac{a^{2l } \, e^{- \frac{a^{2}}{2}}  \Gamma(p+l) p^{1 - 2l} \, \Gamma \left( \frac{m + n + 2l + 1}{2}, \frac{b^{2}}{2} \right) }{  l! \Gamma(n+l+1) 2^{\frac{n - m + 2l + 1}{2}}\,  \Gamma(p-l+1) }, \,\,  \quad  \, \,
\end{equation}

\noindent
where,
 
\begin{equation} \label{Nuttall_Coefficient}
\begin{split}
\mathcal{I}&^{k}_{m, n} (a,b) = \sum_{l = 0}^{m - n + k}  \binom{m - n + k}{l} \frac{(-1)^{k} 2^{\frac{l - 1}{2}} }{a^{n + l - m - k}} \times \\
&\left\lbrace (-1)^{m - n - l - 1} \Gamma \left( \frac{l + 1}{2}, \frac{(b + a)^{2}}{2} \right) + \Gamma\left( \frac{l + 1}{2} \right) \right. \\
& - \left. [{\rm sgn}(b - a)]^{l + 1} \gamma \left( \frac{l + 1}{2}, \frac{(b - a)^{2}}{2} \right)   \right\rbrace,
\end{split}
\end{equation}

\noindent
with 

\begin{equation}
\gamma(a,x) \triangleq \int_{0}^{x} t^{a - 1} \rm{exp} (-t) \, dt  
\end{equation}
denoting the lower incomplete gamma function.
\end{lemma}

\begin{proof}
The truncation error of \eqref{Nuttall_Polynomial} is expressed as,
 
\begin{equation} \label{Nuttall_Truncation_1}
\epsilon_{t} =  \sum_{l = p+1}^{\infty} \frac{a^{2l } \, e^{- \frac{a^{2}}{2}}  \Gamma(p+l) p^{1 - 2l} \, \Gamma \left( \frac{m + n + 2l + 1}{2}, \frac{b^{2}}{2} \right) }{  l! \Gamma(n+l+1) 2^{\frac{n - m + 2l + 1}{2}}\,  \Gamma(p-l+1) },
\end{equation}

\noindent
which can be equivalently expressed as
 
\begin{equation} \label{Tr_e_1}
\begin{split}
\epsilon_{t} &= \underbrace{\sum_{l = 0}^{\infty} \frac{a^{2l } \, e^{- \frac{a^{2}}{2}}  \Gamma(p+l) p^{1 - 2l} \, \Gamma \left( \frac{m + n + 2l + 1}{2}, \frac{b^{2}}{2} \right) }{  l! \Gamma(n+l+1) 2^{\frac{n - m + 2l + 1}{2}}\,  \Gamma(p-l+1) } }_{\mathcal{I}_{1}}\\
& - \sum_{l = 0}^{p} \frac{a^{2l } \, e^{- \frac{a^{2}}{2}}  \Gamma(p+l) p^{1 - 2l} \, \Gamma \left( \frac{m + n + 2l + 1}{2}, \frac{b^{2}}{2} \right) }{  l! \Gamma(n+l+1) 2^{\frac{n - m + 2l + 1}{2}}\,  \Gamma(p-l+1) }.
\end{split}
\end{equation}

\noindent
As already mentioned, since the series in $\mathcal{I}_{1}$ tends to infinity, the terms  $\Gamma(p + l) p^{1 - 2l}{/} \Gamma(p - l + 1)$ vanish, yielding
 
\begin{equation} \label{Tr_e_2}
\mathcal{I}_{1} = \mathcal{Q}_{m, n}(a, b) = \sum_{l = 0}^{\infty} \frac{a^{2l } \, e^{- \frac{a^{2}}{2}}  \, \Gamma \left( \frac{m + n + 2l + 1}{2}, \frac{b^{2}}{2} \right) }{  l! \Gamma(n+l+1) 2^{\frac{n - m + 2l + 1}{2}} }.
\end{equation}

\noindent
It is recalled that the normalized Nuttall $Q$-function can be upper bounded by \cite[eq. (19)]{J:Karagiannidis}, namely,
 
\begin{equation} \label{Nuttall_UB}
\mathcal{Q}_{m, n} (a, b) \leq \mathcal{Q}_{\lceil m\rceil_{0.5}, \lceil n \rceil_{0.5}} (a, b),
\end{equation}

\noindent
where $\lceil n \rceil_{0.5} \triangleq  \lceil n - 0.5 \rceil + 0.5$, with $\lceil . \rceil$ denoting the integer ceiling function. To this effect, by substituting \eqref{Nuttall_UB} into \eqref{Tr_e_2} and then into \eqref{Tr_e_1}, one obtains,
 
\begin{equation} \label{Tr_e_3}
\epsilon_{t}  \leq \mathcal{Q}_{\lceil m\rceil_{0.5}, \lceil n \rceil_{0.5}} (a, b)  - \sum_{l = 0}^{p} \frac{a^{2l } \, e^{- \frac{a^{2}}{2}}  \Gamma(p+l) p^{1 - 2l} \, \Gamma \left( \frac{m + n + 2l + 1}{2},\frac{b^{2}}{2} \right) }{  l! \Gamma(n+l+1) 2^{\frac{n - m + 2l + 1}{2}}\,  \Gamma(p-l+1) }.
\end{equation}

\noindent
Importantly, the upper bound in \eqref{Nuttall_UB} can be expressed in closed-form according to the expression in \cite[Corollary $1$]{J:Karagiannidis}. Therefore, after substitution in \eqref{Tr_e_3} equation \eqref{Nuttall_Truncation} is deduced thus, completing the proof\footnote{By setting $n = m - 1$ in \eqref{Nuttall_Polynomial}, \eqref{Nuttall_Integer_Indices}, and \eqref{Nuttall_Truncation}, similar expressions are deduced for the Marcum $Q$-function, $Q_{m}(a,b)$.}. 
\end{proof}

\begin{remark}

By following the same methodology as in Lemma $1$, a similar upper bound can be deduced for the truncation error of the infinite series of $\mathcal{Q}_{m,n}(a,b)$  in \eqref{Nuttall_Integer_Indices_inf}.

\end{remark}

\subsection{A Tight Upper Bound and Approximation}

Besides the expressions for the special cases that $m + n$ is an odd positive integer and $m \pm 0.5 \in \mathbb{N}, n \pm 0.5 \in \mathbb{N}$, no simple representations exist for the Nuttall $Q$-function.

\begin{proposition}
For $a,b,m,n \in \mathbb{R}^{+}$ and for the special cases that $b \rightarrow 0 \mid a, m, n \geq \frac{3}{2} b$, the following closed-form upper bound for the normalized Nuttall $Q$-function is valid,
 
\begin{equation} \label{Nuttall_Novel_Bound}
\mathcal{Q}_{m,n}(a,b) \leq \frac{ \Gamma\left( \frac{m + n + 1}{2} \right) \,_{1}F_{1} \left( \frac{m + n + 1}{2}, n + 1, \frac{a^{2}}{2} \right)  }{n!\,2^{\frac{n - m + 1}{2}}\,e^{\frac{a^{2}}{2}} }.
\end{equation}

\noindent
where $\,_{1}F_{1}(a, b, x) $ denotes the Kummer confluent hypergeometric function \cite{B:Tables, B:Abramowitz}.

\end{proposition}

\begin{proof}
According to Remark $2$, equation \eqref{Nuttall_Polynomial} reduces to the exact infinite series as $p \rightarrow \infty$. By also recalling that $\int_{0}^{\infty} f(x) dx \geq \int_{a}^{\infty} f(x) dx$ when $a \in \mathbb{R}^{+}$, it follows that $\Gamma(a,x) \leq \Gamma(x)$.  As a result, the $\mathcal{Q}_{m,n}(a,b)$ function can be straightforwardly upper bounded as follows:
 
\begin{equation} \label{Nuttall_Approximation_1}
\mathcal{Q}_{m,n}(a,b)  \leq   \underbrace{\sum_{l = 0}^{\infty} \frac{a^{2l } \, e^{- \frac{a^{2}}{2}}  \, \Gamma \left( \frac{m + n + 2l + 1}{2} \right) }{  l! \Gamma(n+l+1) 2^{\frac{n - m + 2l + 1}{2}} }}_{\mathcal{I}_{2}}.
\end{equation}

\noindent
By recalling that the Pochhammer symbol is defined as $(a)_{n} \triangleq \Gamma(a+n){/}\Gamma(a)$ and expressing each gamma function as 

\begin{equation}
\Gamma(x+l) = (x)_{l} \Gamma(x)
\end{equation}
one obtains,
 
\begin{equation} \label{Nuttall_Approximation_2}
\mathcal{I}_{2} = \frac{\Gamma\left( \frac{m + n + 1}{2} \right) e^{\frac{-a^{2}}{2}}  }{n! 2^{\frac{n - m + 1}{2}} } \sum_{l = 0}^{\infty} \frac{\left( \frac{m + n + 1}{2} \right)_{l} a^{2l}  }{ \, l!\,  (n + 1)_{l} 2^{l} }.
\end{equation}

\noindent
The above infinite series can be expressed in terms of the Kummer's confluent hypergeometric function 

 \begin{equation}
 \,_{1}F_{1} (a,b,x) \triangleq \sum_{i = 0}^{\infty} \frac{(a)_{i}}{(b)_{i}} \frac{x^{i}}{i!}
 \end{equation}
  
    Hence, by performing the necessary change of variables and substituting \eqref{Nuttall_Approximation_2} into \eqref{Nuttall_Approximation_1}, equation \eqref{Nuttall_Novel_Bound} is deduced and thus the proof is completed.
\end{proof}

\begin{figure}[h!]
\centering
\includegraphics[ width=12cm,height=9cm]{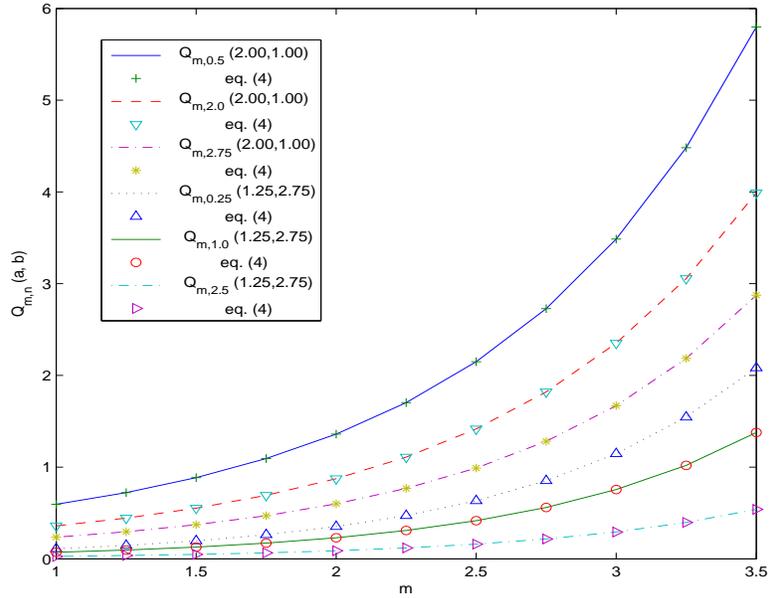}
\caption{ $Q_{m,n}(a,b)$ in \eqref{Nuttall_Polynomial}.}
\end{figure}

\begin{remark}
When $a, m, n \geq \frac{5}{2} b$, the upper bound in \eqref{Nuttall_Novel_Bound} becomes an accurate closed-form approximation for $\mathcal{Q}_{m,n}(a,b)$, namely, 
 
\begin{equation} \label{Nuttall_Novel_Approx}
\mathcal{Q}_{m,n}(a,b) \simeq\frac{ \Gamma\left( \frac{m + n + 1}{2} \right) \,_{1}F_{1} \left( \frac{m + n + 1}{2}, n + 1, \frac{a^{2}}{2} \right)  }{n!\,2^{\frac{n - m + 1}{2}}\,e^{\frac{a^{2}}{2}} }.
\end{equation}

\end{remark}

\noindent
The Nuttall $Q{-}$function is neither tabulated nor built-in in popular mathematical software packages like MAPLE, MATLAB and MATHEMATICA. Hence, the derived expressions are rather useful both analytically and numerically\footnote{Similar expressions for the $Q_{m,n}(a,b)$ function can be straightforwardly deduced with the aid of the identity $ \mathcal{Q}_{m,n}(a,b) = Q_{m,n}(a,b) {/}a^{n}$ i.e. by multiplying equations \eqref{Nuttall_Polynomial}, \eqref{Nuttall_Integer_Indices}, \eqref{Nuttall_Truncation}  and \eqref{Nuttall_Approximation_1} with $a^{n}$.}.

The behaviour of \eqref{Nuttall_Polynomial} is illustrated in Fig. $1 $ along with results obtained from numerical integrations. The series was truncated after $20$ terms and one can notice the excellent agreement between analytical and numerical results. This is also verified through the level of the involved absolute relative error $\epsilon_{r} \triangleq \mid \mathcal{Q}_{m,n}(a,b) - \tilde{\mathcal{Q}}_{m,n}(a,b) \mid {/} \mathcal{Q}_{m,n}(a,b) $ which is less than $\epsilon_{r} < 10^{-11}$. Similarly, the behaviour of \eqref{Nuttall_Novel_Bound} is depicted in Fig. $2$. Clearly, it upper bounds the $\mathcal{Q}_{m,n}(a,b)$ tightly while it becomes an accurate approximation for higher values of $a$.

\begin{figure}[h!]
\centering
\includegraphics[ width=12cm,height=9cm]{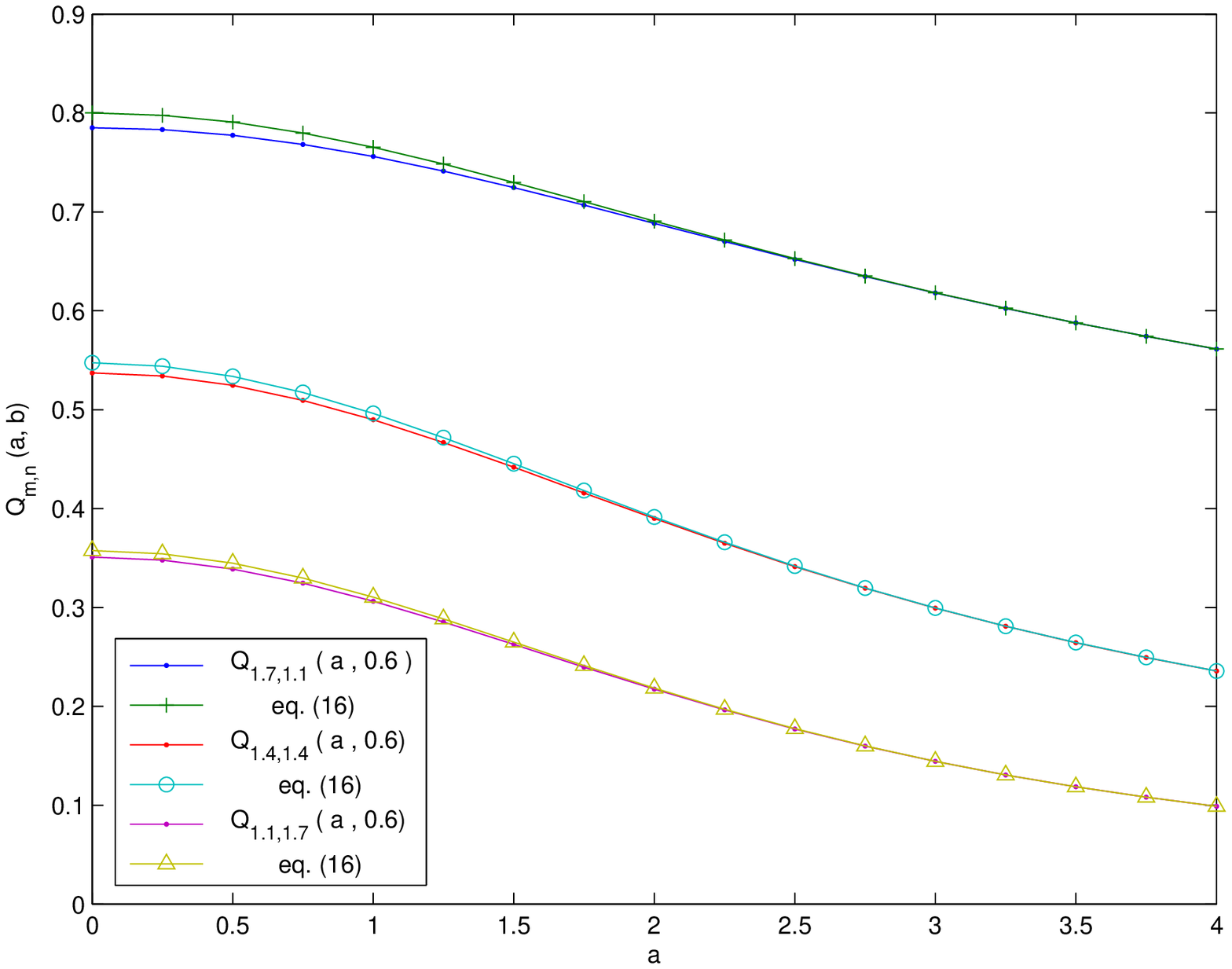}
\caption{ $Q_{m,n}(a,b)$ in \eqref{Nuttall_Novel_Bound}.}
\end{figure}

\section{Novel Analytic Results for the Incomplete Toronto Function }

\subsection{Definition and Basic Properties}

The ITF has been also useful in telecommunications and is defined as,
 
\begin{equation} \label{ITF_Definition}
T_{B}(m,n,r) \triangleq 2r^{n-m+1} e^{-r^{2}} \int_{0}^{B} t^{m-n} e^{-t^{2}}I_{n}(2rt) dt.
\end{equation}

\noindent
When $B = \infty$, the $T_{B}(m,n,r)$ function reduces to the $T(m, n, r)$ function, while for the specific case  $n = (m-1){/}2$ it is expressed in terms of the Marcum $Q$-function, namely,
 
\begin{equation} \label{ITF_Marcum}
T_{B}\left(m,\frac{m-1}{2},r\right) = 1 - Q_{\frac{m+1}{2}}\left(r\sqrt{2},B\sqrt{2}\right).
\end{equation}
\noindent
Alternative representations to the $T_{B}(a,b,r)$ function include two series which are infinite and no study has been reported on their convergence and truncation  \cite{J:Sagon}. In addition, the following polynomial representation was proposed in \cite{C:Sofotasios_1},
 
\begin{equation} \label{ITF_Series}
T_{B}(m,n,r) \simeq \sum_{k = 0}^{p} \frac{\Gamma(p + k) r^{ 2(n + k) - m + 1} \gamma\left( \frac{m + 1}{2} + k, B^{2} \right) }{k! p^{2k - 1}\Gamma(p - k + 1) \Gamma(n + k + 1) e^{r^{2}} }.
\end{equation}
which as $p \rightarrow \infty$ it reduces to,
 
\begin{equation} \label{ITF_Series}
T_{B}(m,n,r) = \sum_{k = 0}^{\infty} \frac{  r^{ 2(n + k) - m + 1} \gamma\left( \frac{m + 1}{2} + k, B^{2} \right) }{k!    \Gamma(n + k + 1) e^{r^{2}} }.
\end{equation}

\subsection{A Closed-Form Upper Bound for the Truncation Error}

A tight upper bound for the truncation error of \eqref{ITF_Series} can be derived in closed-form.

\begin{lemma}
For  $m, n, r \in \mathbb{R}$, $B \in \mathbb{R}^{+}$ and $m > n$ the following closed-form inequality holds,
 
\begin{equation} \label{ITF_Truncation_1}
\begin{split}
\epsilon_{t} &\leq \sum_{k = 0}^{\lfloor n \rfloor _{0.5} - \frac{1}{2}} \sum_{l = 0}^{L} \frac{ r^{-(2k +l)} \, \left(\lfloor n \rfloor _{0.5} + k - \frac{1}{2}\right)! \, (L - k)!  }{\, k! \, l!\left( \lfloor n \rfloor _{0.5} - k - \frac{1}{2} \right)! (L-k-l)!}\\
& \times \left\lbrace \frac{\gamma \left[ \frac{l + 1}{2}, (B + r)^{2} \right]}{(-1)^{\lceil m \rceil  - l} \,2^{2k+1}} + \frac{\gamma \left[ \frac{l + 1}{2}, (B- r)^{2} \right]}{(-1)^{k} \,  2^{2k+1}} \right\rbrace\\
& - \sum_{k = 0}^{p} \frac{\Gamma(p + k) r^{ 2(n + k) - m + 1} \gamma\left( \frac{m + 1}{2} + k, B^{2} \right) }{k! p^{2k - 1}\Gamma(p - k + 1) \Gamma(n + k + 1) e^{r^{2}} }.
\end{split}
\end{equation}
\end{lemma}

\begin{proof}
Since the corresponding truncation error is expressed as 

\begin{equation}
\epsilon_{t} = \sum_{p + 1}^{\infty} f(x) = \sum_{l = 0}^{\infty} f(x) - \sum_{l}^{p} f(x)
\end{equation}

and given that \eqref{ITF_Series} reduces to an exact infinite series as $p \rightarrow \infty$, it follows that,
 
\begin{equation} \label{ITF_Truncation_2}
\epsilon_{t}  = \underbrace{\sum_{k = 0}^{\infty} \frac{r^{ 2(n + k) - m + 1} \gamma\left( \frac{m + 1}{2} + k, B^{2} \right) }{k!  \Gamma(n + k + 1) e^{r^{2}} } }_{\mathcal{I}_{3}} - \sum_{k = 0}^{p} \frac{\Gamma(p + k) r^{ 2(n + k) - m + 1} \gamma\left( \frac{m + 1}{2} + k, B^{2} \right) }{k! p^{2k - 1}\Gamma(p - k + 1) \Gamma(n + k + 1) e^{r^{2}} }.
\end{equation}

\noindent
Given that $\mathcal{I}_{3} = T_{B}(m, n, r)$,   the $\epsilon_{t}$ can be upper bounded as follows:
 
\begin{equation} \label{ITF_Truncation_3}
\epsilon_{t}   \leq T_{B}(\lceil m \rceil, \lfloor n \rfloor_{0.5}, r) - \sum_{k = 0}^{p} \frac{\Gamma(p + k) r^{ 2(n + k) - m + 1} \gamma\left( \frac{m + 1}{2} + k, B^{2} \right) }{k! p^{2k - 1}\Gamma(p - k + 1) \Gamma(n + k + 1) e^{r^{2}} }.
\end{equation}

\noindent
The $ T_{B}(\lceil m \rceil, \lfloor n \rfloor_{0.5}, r)$ function can be expressed in closed-form according to the closed-form expression in \cite{C:Sofotasios_1}. Hence, by substituting in \eqref{ITF_Truncation_3} equation \eqref{ITF_Truncation_1} is obtained, which completes the proof.
\end{proof}

\begin{remark}
By omitting the terms $\Gamma(p + k) p^{1 - 2k} {/} \Gamma(p - k + 1)$ in \eqref{ITF_Truncation_1}, a closed-form upper bound can be  deduced for the truncation error of the exact infinite series in Remark $2$.

\end{remark}

\subsection{A Tight Closed-form Upper Bound and Approximation}

Capitalizing on the algebraic form of the $T_{B}(m, n, r)$ function, a simple closed-form upper bound is derived which in certain cases becomes an accurate approximation.

\begin{proposition}
For $m, n, r \in \mathbb{R}$, $B \in \mathbb{R}^{+}$ and $m, n, r \leq \frac{B}{2}$, the following inequality holds,
 
\begin{equation} \label{ITF_Appr_1}
T_{B}(m, n, r) \leq  \frac{ \Gamma\left( \frac{m + 1}{2} \right)  \, _{1}F_{1} \left( \frac{m + 1}{2}, n + 1, r^{2} \right)    }{  r^{m - 2n - 1  }     \Gamma(n + 1) e^{r^{2}} }.
\end{equation}
\end{proposition}

\begin{proof}
It is recalled that the $\gamma(a,x)$ function can be upper bounded by the $\Gamma(a)$ function since 

 \begin{equation}
 \Gamma(a) = \int_{0}^{\infty} t^{a-1}{\rm exp}(-t) dt = \int_{0}^{x = \infty} t^{a-1}{\rm exp}(-t) dt = \gamma(a, x = \infty)
  \end{equation}
   
     Therefore, the $T_{B}(m,n,r) $ function can be upper bounded as,
 
\begin{equation} \label{ITF_Appr_2}
T_{B}(m,n,r) \leq  \sum_{k = 0}^{\infty} \frac{r^{ 2(n + k) - m + 1} \Gamma\left( \frac{m + 1}{2} + k \right) }{k! \Gamma(n + k + 1) e^{r^{2}} },
\end{equation}

\noindent
which with the aid of the identity

\begin{equation}
\Gamma(a,n) = (a)_{n}\Gamma(a)
\end{equation}

can be expressed as,
 
\begin{equation} \label{ITF_Appr_3}
T_{B}(m,n,r) \leq   \frac{ r^{2n - m + 1}  \Gamma \left(\frac{m + 1}{2} \right)}{\Gamma(n + 1) \,e^{r^{2}}} \sum_{k = 0}^{\infty} \frac{\left( \frac{m + 1}{2} \right)_{k} }{(n + 1)_{k}  } \frac{r^{2k}}{k!}.
\end{equation}

\begin{figure}[h!]
\centering 
\includegraphics[ width=15cm,height=12cm]{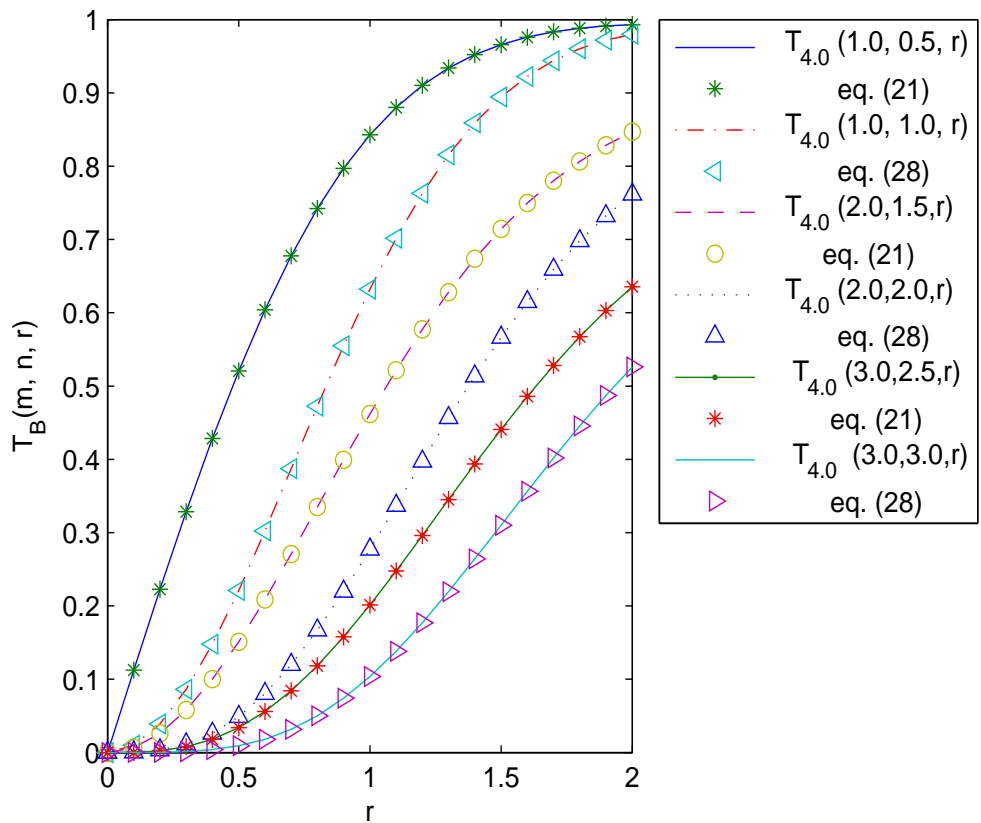}
\caption{$T_{B}(m,n,r)$ in (23) }
\end{figure}

\noindent
Importantly, the above series can be expressed in closed-form in terms of the Kummer's confluent hypergeometric function. Thus, equation \eqref{ITF_Appr_1} is deduced and the proof is completed.
\end{proof}

\begin{remark}
The proof can be also completed by assuming $B \rightarrow \infty$ in \eqref{ITF_Definition}  and utilizing \cite[eq. (8.406.3)]{B:Tables} and \cite[eq. (6.631.1)]{B:Tables}. The incomplete Toronto function reduces then to the Toronto function, as the two functions are related by the identity $T_{B = \infty}(m, n, r) = T(m, n, r)$. Furthermore \eqref{ITF_Appr_1}  becomes an accurate approximation when $m, n, r \leq 2B$, namely, 
 
\begin{equation} \label{ITF_Appr_5}
T_{B}(m, n, r) \simeq  \frac{ \Gamma\left( \frac{m + 1}{2} \right)  \, _{1}F_{1} \left( \frac{m + 1}{2}, n + 1, r^{2} \right)    }{  r^{m - 2n - 1  }     \Gamma(n + 1) e^{r^{2}} }.
\end{equation}

\end{remark}

Equation   \eqref{ITF_Series} is depicted in Fig.$3{\rm }$ along with results from corresponding numerical integrations.The agreement between analytical and numerical results is excellent and the relative error for \eqref{ITF_Series} is $\epsilon_{r} < 10^{-4}$ for truncation after $20$ terms. In the same context  (28) is depicted in Fig.$4$ along with numerical results for three different scenarios. The involved relative error is proportional to the value of $r$ and is $\epsilon_{r} < 10^{-6}$ when $r < 1$.

\begin{figure}[h!]
\centering 
\includegraphics[ width=12cm,height=9cm]{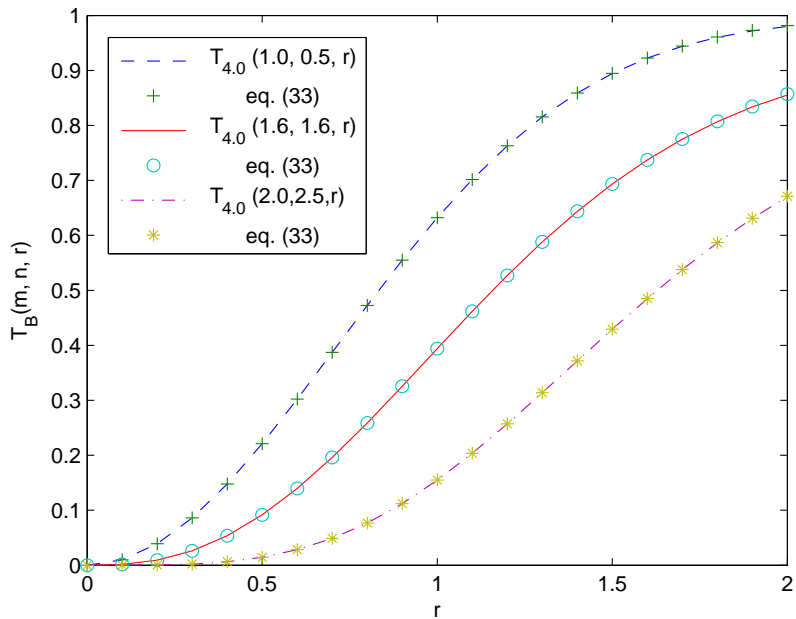}
\caption{$T_{B}(m,n,r)$ in  (28) }
\end{figure}

To the best of the Authors' knowledge, the offered results have not been reported in the open technical literature and can be particularly useful in emerging wireless technologies such as cognitive radio, cooperative communications, MIMO systems, digital communications over fading channels, ultasound and free-space-optical communications, among others \cite{B:Sofotasios, C:Sofotasios_1, C:Sofotasios_2, C:Sofotasios_3, J:Sofotasios, P:Paschalis_1, P:Paschalis_2, P:Paschalis_3, J:Theo, P:Paschalis_4,  P:Paschalis_5, Paschalis_4, Paschalis_6, Paschalis_7, Paschalis_8, Paschalis_9, Paschalis_10, Paschalis_11, Paschalis_12, Paschalis_13, Paschalis_14, Paschalis_15, Paschalis_16, Paschalis_17}, and the references therein.

\section{Conclusion}

Novel  upper bounds were derived for the truncation error of representations for the Nuttall $Q$-function and the incomplete Toronto function. These expressions are given in closed-form and are have a tractable algebraic form. Since the Nuttall$-Q$ and incomplete toronto functions are not included as built-in functions in popular mathematical software packages, the offered results are useful in computing these functions efficiently. As a result, the accurate computation of critical performance measures in digital communications that involve these functions becomes more feasible.

\bibliographystyle{IEEEtran}

\end{document}